\documentclass[11pt]{article}
\pdfoutput=1
\usepackage[centertags]{amsmath}
\usepackage{amssymb}

\newtheorem{theorem}{Theorem}[section]

\newtheorem{lemma}{Lemma}[section]

\providecommand{\expectation}[2]{\mathbb{E}_{#2}\left[#1\right]}
\providecommand{\probab}[2]{\mathbb{P}_{#2}\left[#1\right]}

\providecommand{\binom}[2]{{#1\choose#2}}

\newenvironment{proof}[0]{\textit{Proof.} }{\hfill  $\blacksquare$ } 

\providecommand{\degree}[2]{{\textrm{deg}_{#1}(#2)}}
\providecommand{\vcprob}[3]{f^{#1}_{#2}(#3)}

\providecommand{\rktree}[2]{G^{#1}(#2)}

\providecommand{\tdeg}[2]{X_{#1}(#2)}

\providecommand{\edeg}[2]{\expectation{\tdeg{#1}{#2}}{}}

\providecommand{\ktree}{k-tree}
\providecommand{\ktrees}{k-trees}

\begin{document}

\title{The Degree Distribution of Random k-Trees}

\author{Yong Gao \thanks{Supported in part by NSERC Discovery Grant RGPIN 327587-06} \\
    Department of Computer Science, \\
    Irving K. Barber School of Arts and Sciences, \\
    University of British Columbia Okanagan, \\
    Kelowna, Canada V1V 1V7
}
\maketitle

\begin{abstract}
A power law degree distribution is established for
a graph evolution model based on the graph class of \ktrees. This \ktree-based
graph process
can be viewed as an idealized model that captures some characteristics of
the preferential attachment and copying mechanisms that existing evolving graph processes
fail to model due to technical obstacles. The result also serves as a further cautionary note
reinforcing the point of view that a power law degree distribution
should not be regarded as
the only important characteristic of a complex network, as has been previously argued \cite{dimitris05power,li05,mitzenmacher05}.

\end{abstract}

\section{Introduction}

Since the discovery of the power-law degree distribution of the web graphs and other complex large-scale networks, many random models
for such networks have been proposed \cite{albert02complex,bollobas01scalefree,cooper03,dorogovtsev02}.
By studying a variety of graph models with a power law degree distribution, it is hoped that one can gain insight into the characteristics
of real-world complex networks that are algorithmically exploitable, and can use these models as a tool for empirical studies \cite{chakrabarti06}.
It is therefore desirable to have random models that not only exhibit power law degree distributions, but also have other
structural features specified in a controlled manner.

Most of the existing models for complex networks
define a graph evolution process in which vertices are added to the current graph one at a time.
In each time step, the newly-added vertex is connected to
a number of existing vertices selected according to some probability distribution.
Two popular ways to specify the probability distribution for vertex selection are \textit{preferential attachment}
and \textit{copying}
(also known as \textit{duplication}). In the preferential attachment model, an existing vertex is selected with probability
in proportion to its vertex degree. In the copying model, neighbors of an existing vertex (selected uniformly at random) are
sampled to determine the vertices to connect to.

Bollobas et al. \cite{bollobas01scalefree}  proved the first rigorous result on the power law degree distribution of such
graph evolution models, showing that with high probability the degree distribution of the Barabasi-Albert model \cite{albert02complex}
obeys a power law $d^{-3}$.
Since then, many variants of the preferential attachment model have been proposed
by introducing additional parameters that manipulate the probability with which an existing vertex is to be selected. The motivation
is to construct models that obey a power law degree distribution with the exponent depending on some adjustable parameters so that
a variety of power law distributions observed in the real-world setting can be modelled.
Jordan \cite{jordan06} analyzed a slightly generalized model investigated by Dorogovtsev et al. \cite{dorog00}
and showed that for large constant $d > 0$, the proportion of vertices of  degree $d$
follows a power law $d^{-\gamma}$ with the exponent $\gamma \in (2, \infty)$ determined by two adjustable parameters.
Aiello, Chung, and Lu \cite{aiello01} and Cooper and Frieze \cite{cooper03} studied even more general
preferential-attachment  models with a set of parameters. These parameters specify the number of existing vertices to be selected in each step and control in a
probabilistic way how these vertices are selected. A vertex can be selected by sampling uniformly at random from existing vertices or by the preferential attachment mechanism.
Among the other results, Cooper and Frieze showed that in their general model the proportion of vertices of degree $d > 0$ follows a power law with the
exponent $\gamma \in (2, \infty)$ determined by the model parameters. In all of the preferential attachment
models, it is an essential assumption that the vertices to be connected to the new vertex are selected independently of each other.

The first model with copying mechanism for the web graphs is proposed in \cite{kumar}.  A similar model, called the duplication model, arises in the context of biological networks \cite{chung03}.
With the copying mechanism,  a new vertex $v_{n + 1}$ is connected
to a set of existing vertices using the following scheme:
\begin{enumerate}
\item An existing vertex $v_i$ is selected uniformly at random to copy from.
\item Let $N(v_i)$ be the set of neighbors of $v_i$ in $\{v_1, v_2, \cdots, v_{i - 1}\}$.
The vertex $v_{n + 1}$ is then connected to a subset of $N(v_i)$ selected in a probabilistic fashion.
The number of neighbors that $v_{n+1}$ is connected to is called the out-degree of $v_{n+1}$.
\end{enumerate}
Without any extra work, the above copying mechanism generates a star-like graph centered on the initial graph $G_0$.
To overcome this limitation, Kumar et al. \cite{kumar} require that the out-degree (i.e., the number of out-edges) is a constant
and implement this by connecting the new vertex to either its neighbors or other vertices selected uniformly at random  which is crucial for the construction to work.
For the case that the out-degree is 1, it was proved in \cite{kumar} that the in-degree sequence has a power law distribution with high probability.

In the duplication models studied in \cite{chung03,bebek06}, $N(v_i)$ is extended to contain all the neighbors of $v_i$ and each vertex in this extended $N(v_i)$ is
connected to $v_{n + 1}$ independently with a certain probability. As noted in \cite{bebek06}, a correction step has to be employed to avoid the generation
of degenerate graph processes.  Power law distributions for the expected fraction of vertices of a given degree are proved
in \cite{chung03,bebek06}. 
Cooper and Frieze \cite{cooper03} use a copying scheme in which the neighbors of $v_{n+1}$ is selected one at a time
by repeating the process a number of times independently. This makes the (highly-complicated) analysis
more approachable, but spoils to a large extent the idea of
the copying mechanism that is intended to capture the phenomenon that neighboring vertices are likely to be connected together to a new vertex.

In this paper, we study a random model for the well-known graph class of \ktrees,
which may serve as an alternative (and idealized) model in the study of complex networks.
The notion of k-trees is a generalization of trees and is closely related to the concept of treewidth in graph theory \cite{kloks94}.
We show that the degree distribution of a graph evolution process obtained by a straightforward randomization of the recursive definition of \ktrees\ obeys the power law
$$
d^{-(1 + \frac{k}{k-1})}
$$
with high probability for large $d$, where $k$ is the parameter that characterizes the degree to which a graph is tree-like.
In addition to introducing an alternative model with preferential attachment and copying mechanisms, we hope that
the fact that a power law degree distribution exists in such a graph class with quite unique structural characteristics
serves as a further cautionary note, reinforcing the viewpoint that a power law degree distribution should not be regarded as
the only important characteristic of a complex network, as has been previously argued in \cite{li05,mitzenmacher05}. We note that in \cite{dimitris05power}, 
the inherent bias of existing approaches in the empirical study of 
the Internet graph was identified --- it was shown that the widely-used 
traceroute sampling method ``can make power laws appear where none existed in
the underlying graphs!"

In the next section, we introduce the construction of the random \ktrees and discuss its relation to existing models of complex networks.
In Section 3, we prove the power law degree distribution of random \ktrees.
We conclude in Section 4 
with a discussion on the construction of random partial \ktrees.

\section{Random k-Trees: the Construction}

Throughout this paper, the degree of a vertex $v$ in a graph $G$ is denoted by $\degree{G}{v}$.  A $k$-clique of a graph is understood as a complete subgraph on a set of $k$ vertices. All the graphs considered in this paper are undirected.

The construction of a random \ktree\ is based on the following simple randomization of the recursive definition of \ktrees\ \cite{kloks94}.
Starting with an initial clique $\rktree{k}{k + 1}$ of size $k + 1$, a sequence of graphs $\{\rktree{k}{n}, n \geq k + 1\}$ is constructed by adding vertices to the graph one at a time.
To construct $\rktree{k}{n + 1}$, we add a new vertex $v_{n + 1}$ and then connect it to the $k$ vertices of a k-clique selected
uniformly at random from all the k-cliques in $\rktree{k}{n}$.
We call the graph process $\{\rktree{k}{n}, n \geq k + 1\}$ a \textit{k-tree process}.

\subsection{Relations to Existing Models}
In this subsection, we discuss some basic properties of the $k$-tree process, 
including the number of $k$-cliques in $G^{k}(n)$ and the probability 
that an existing vertex of a given degree is connected to a new vertex. These properties 
are needed in the proof of  our main result. They also enable us to illustrate further the relations between the $k$-tree process and existing graph evolution models.     

Let $\mathcal{C}_{n}$ be the set of cliques of size $k$ in the graph $\rktree{k}{n}$.
It is easy to see that when a new vertex is added, exactly $\binom{k}{k - 1}$ new $k$-cliques are created and none of the existing $k$-cliques is destroyed.
So, taking into consideration the initial clique of size $k + 1$, we see that the total number of $k$-cliques in
$\rktree{k}{n}$ is
$$
|\mathcal{C}_n| = (n - k - 1)k + (k + 1).
$$

Consider a vertex $v$ in $G^{k}(n)$. 
Since every time a new vertex is added and connected to the vertex $v$, 
exactly $\binom{k - 1}{k - 2}$  new $k$-cliques are created
that contain $v$ as one of its vertices, the total number of $k$-cliques in $\rktree{k}{n}$ containing $v$ is
\begin{equation}
 n^{*} = \binom{k}{k - 1} + \binom{k - 1}{k - 2}(\degree{\rktree{k}{n}}{v} - k),
\end{equation}
where the first term is the number of the $k$-cliques containing $v$ that are created when $v$ is added to the graph and the second term is the total number of $k$-cliques 
containing $v$ that are created later on when $v$ is connected to new vertices. 

Therefore given $G^{k}(n)$ (i.e., conditional on $G^{k}(n)$),  the conditional
probability for $v$ to be connected to the new vertex $v_{n + 1}$ is
\begin{eqnarray}
\label{eq:prop}
&&\probab{v \textrm{ is connected to } v_{n + 1}\ |\ \rktree{k}{n}}{} \nonumber \\ 
 &=& \frac{n^{*}}{|\mathcal{C}_n|}
 = \frac{a_k\degree{\rktree{k}{n}}{v} - b_k}{c_k n}
\end{eqnarray}
where $a_k = k - 1, b_k = k(k-2), \textrm{ and } c_k = k - \frac{k^2 - 1}{n}$. 

Note that the above expression only depends on the degree of $v$ in $\rktree{k}{n}$. 
It follows that, given $\degree{\rktree{k}{n}}{v} = d$,
the conditional probability for $v$ to be connected to $v_{n + 1}$ is
\begin{equation}
\label{eq-attach-prob}
\vcprob{k}{d}{n} \stackrel{\text{def}}{=} 
  \probab{v \textrm{ is connected to } v_{n+1} \ |\ \degree{\rktree{k}{n}}{v} = d}  {}
 =  \frac{a_kd - b_k}{c_k n}
\end{equation}
where $a_k = k - 1, b_k = k(k-2), \textrm{ and } c_k = k - \frac{k^2 - 1}{n}$.
We see that even though there is no explicit preferential-attachment mechanism employed, equation (\ref{eq-attach-prob}) shows that the construction scheme does have a similar effect.

\subsection{The Advantages of Random $k$-Trees}
The \ktree\ construction scheme can be viewed as a very rigid copying mechanism; In step $n + 1$, the new vertex $v_{n + 1}$
is connected to an existing vertex $v_i$ selected uniformly at random from $\{v_1, v_2, \cdots, v_n\}$ and to a subset of $k - 1$ vertices selected uniformly at random
without replacement from the $k$ neighbors that are connected to $v_i $ in step $i$.

As has been discussed in Section 1, in almost all the existing  preferential-attachment models
and copying models, there is an essential assumption that old vertices to be connected to 
a new vertex are selected independently. The random $k$-tree model studied in the current paper
is unique in that  these vertices are selected in a highly correlated manner. This
captures in a better way the phenomenon that neighboring vertices are more likely to be 
connected to a new vertex, which is exactly what the copying mechanism tries to model.  
In addition, the random $k$-tree has by construction a treewidth $k$ --- a structural
feature of algorithmic significance that none of the existing models  
has a mechanism to control.

\section{The Degree Distribution of Random $k$-Trees}
This section is devoted to proving that for the \ktree\ process, the proportion of vertices of degree $d$ follows asymptotically
a power law $d^{-\gamma}$ with  exponent $\gamma = 1 + \frac{k}{k - 1}$.  Throughout the discussion, we assume that
$k$ is a fixed constant. 
In the following, we use $\tdeg{d}{n}$ to denote the random variable for the total number
of vertices of degree $d$ in $\rktree{k}{n}$, and write    
$$
\alpha_{d} \triangleq
      \frac{\Gamma(3 + \frac{2}{k - 1})}{\Gamma(1 + \frac{1}{k - 1})} \frac{\Gamma(d - \frac{k(k-2)}{k-1})}{\Gamma(d - \frac{k(k-2)}{k-1} + \frac{k}{k-1} + 1)}
$$
which, by Stirling's approximation, is approximately
$$
e^{-(1 + \frac{k}{k-1})}d^{-(1 + \frac{k}{k-1})}
$$
for large $d$.

Denote by $\mathcal{F}_n = \sigma(\rktree{k}{n}, n \geq 1)$
the $\sigma$-algebra generated by the \ktree\ process up to time $n$. 
We use $I_{A}$ to denote the indicator function of an event $A$. To ease the presentation, we 
use $I_d(i, n)$ to denote the indicator function of the event that the degree of the vertex 
$v_i$ in $\rktree{k}{n}$ is $d$, i.e.,
$$
I_d(i, n) = \left \{
\begin{array}{ll}
1, & \ \ \degree{\rktree{k}{n}}{v_i} = d \\
0, & \ \ \textrm{otherwise.}
\end{array}
\right.
$$

The following simple observation will be used to deal with the case $d = k$.
\begin{lemma}
\label{lem:base}
For any vertex $v$ and $n \geq k + 1$, $\degree{\rktree{k}{n}}{v} \geq k$. Furthermore, for $n \geq k + 2 $
any k-clique in $\rktree{k}{n}$ contains at most one vertex with $\degree{\rktree{k}{n}}{v} = k$.
\end{lemma}
\begin{proof}
The first claim that $\degree{\rktree{k}{n}}{v} \geq k$ follows from the fact that when a new vertex is added, 
it is connected to the $k$ vertices of the selected $k$-clique.  

We use induction to prove the second claim. First, consider the base case of $n = k + 2$. Recall that $\rktree{k}{k + 2}$ 
is obtained by connecting a new vertex $v_{k + 1}$ to the vertices of a $k$-clique 
in the initial $(k + 1)$-clique. We see that in $\rktree{k}{k + 2}$ there are exactly two vertices of degree $k$, namely
the vertex $v_{k+1}$ and one of the vertices in $\{v_1, \cdots, v_k\}$ that is not connected to $v_{k + 1}$.
Therefore, no $k$-clique in $\rktree{k}{k + 2}$ contains more than one vertex of degree $k$, and thus the second claim 
holds for the base case of $n = k + 2$. 

Assume that the second claim holds for $\rktree{k}{n}$. Consider the graph $\rktree{k}{n + 1}$ obtained from $\rktree{k}{n}$. 
Note that by adding a new vertex $v_{n + 1}$ to $\rktree{k}{n}$ and connecting it to the vertices of a $k$-clique in 
$\rktree{k}{n}$, exactly $k$ new $k$-cliques are created each of which has $v_{n+1}$ as its only vertex of degree $k$. By the assumption that the second claim holds for $\rktree{k}{n}$,  
no $k$-clique in $\rktree{k}{n + 1}$ contains more than one
vertices of degree $k$. This completes the induction step and the second claims follows.             
\end{proof}

The next theorem shows that the expected degree sequence of the \ktree\ process obeys a power law distribution.
\begin{theorem}
\label{theorem-average}
Let $\edeg{d}{n}$ be the expected number of vertices with degree $d$ in the random \ktree\ $\rktree{k}{n}$. There exists a
constant $N = N(k)$ (independent of $d$) such that for any $n > N$,
\begin{equation}
\label{eq-average-limit}
\left|\edeg{d}{n} - \alpha_dn\right| \leq C
\end{equation}
where $C = C(k)$ is a constant that is independent of $d$ and $n$.
\end{theorem}

The above result is proved by first establishing a recurrence  
for the expected number $\edeg{d}{n}$ of vertices with a given degree, and then showing that 
$\edeg{d}{n}$ can be asymptotically approximated by $\beta_d n$ where
the sequence $\{\beta_d\}$ is the unique solution to the following simple recurrence relation 
\begin{equation}
\label{eq-limit-case}
 \beta_d = \frac{a_k(d - 1) - b_k}{a_kd - b_k + k} \beta_{d - 1},\ \  \beta_{k} = \frac{1}{2}.
\end{equation}

Recall that to construct the graph $\rktree{k}{n + 1}$ from $\rktree{k}{n}$,
a new vertex added to the graph  will be connected to all the vertices of
a randomly-selected k-clique. This creates a high correlation between the degree of the vertices.
A recurrence is still possible due to the fact that
the conditional probability for a vertex $v$ to have
a degree $d$ in $\rktree{k}{n + 1}$ given $\rktree{k}{n}$ only depends on
the degree of $v$ in $\rktree{k}{n}$. A detailed account is given in the following proof.

\begin{proof}\ [\textbf{Proof of Theorem \ref{theorem-average}}]
To begin with, consider the base case $d = k$.
Due to Lemma~\ref{lem:base}, we have
\begin{equation}
\tdeg{k}{n+1} = \left \{
   \begin{array}{ll}
    & \tdeg{k}{n}, \textrm{ if a k-clique containing a vertex of degree k is selected}, \\
    & \tdeg{k}{n} + 1, \textrm{ otherwise}
   \end{array}
   \right.
\end{equation}
Let $A$ be the event that a $k$-clique containing a vertex of degree $k$ is selected
in step $n + 1$ and let $I_{A}$ be its indicator function. We have
$$
  X_{k}(n + 1) = X_{k}(n)I_{A} + (X_{k}(n) + 1)I_{A^c} 
$$
where $A^c$ is the complement of $A$.

By Lemma~\ref{lem:base}, a $k$-clique contains at most 
one vertex of degree $k$. It follows that the conditional expectation
of $I_A$ (which is equal to the conditional probability of $A$) 
is equal to $X_{k}(n)$, the  total number of degree-$k$ vertices
in $\rktree{k}{n}$, times the conditional probability that a degree-$k$
vertex is selected to be connected to $v_{n+1}$, i.e.,    
\begin{equation}
\label{eq:condp-base}
  \expectation{I_A | \mathcal{F}_n}{} = \vcprob{k}{k}{n} X_k(n). 
\end{equation}

Therefore, by the basic  properties of conditional expectation in theory of probability,
we have
\begin{eqnarray}
\label{eq:conde-base}
\expectation{\tdeg{k}{n + 1}| \mathcal{F}_n}{} &=& 
    \expectation{\tdeg{k}{n}I_A + (\tdeg{k}{n} + 1) I_{A^c} | \mathcal{F}_n}{} \nonumber \\
 &=& \expectation{\tdeg{k}{n}I_A | \mathcal{F}_n}{} +    
        \expectation{(\tdeg{k}{n} + 1) I_{A^c} | \mathcal{F}_n}{}  \nonumber \\
 &=& \expectation{I_A | \mathcal{F}_n}{}\tdeg{k}{n}
         + \expectation{I_{A^c} | \mathcal{F}_n}{} (\tdeg{k}{n} + 1)  
\end{eqnarray}
where the last equality is due to the fact that $X_{k}(n)$ is measurable
with respect to $\mathcal{F}_n$ (i.e., in the context of discrete probability space, $X_{k}(n)$
is a function of $G^{k}(n)$).

Combining equation (\ref{eq:condp-base}) and equation (\ref{eq:conde-base}), we have
\begin{eqnarray}
\label{eq:conde-base1}
\expectation{\tdeg{k}{n + 1}| \mathcal{F}_n}{} &=& \vcprob{k}{k}{n}\tdeg{k}{n} \tdeg{k}{n} + (1 - \vcprob{k}{k}{n}\tdeg{k}{n})(\tdeg{k}{n} + 1) \nonumber\\
              &=& 1 + (1 - \vcprob{k}{k}{n}) \tdeg{k}{n}. 
\end{eqnarray}
By the mathematical definition, $\expectation{\tdeg{k}{n + 1}| \mathcal{F}_n}{}$ itself
is a random variable measurable with respect to $\mathcal{F}_n$. 
Recall, from the probability theory, that the unconditional
expectation of the conditional expectation of a random variable is equal to the 
unconditional expectation of the random variable itself. So, we have
$$
\expectation{\expectation{\tdeg{k}{n + 1}| \mathcal{F}_n}{}}{} = \edeg{k}{n + 1}.
$$
Therefore, by taking expectations on both sides of equation
(\ref{eq:conde-base1}), we get the following recurrence 
\begin{equation}
\label{eq:conde-baser} 
\edeg{k}{n + 1}  = 1 + (1 - \vcprob{k}{k}{n}) \edeg{k}{n}.
\end{equation}

Solving the above recurrence (\ref{eq:conde-baser}) with $\edeg{k}{k+2} = 2$ gives us
\begin{equation}
\label{eq-formula-base}
\edeg{k}{n} =  \frac{1}{2}n + O(1).
\end{equation}

We now consider the general case of $d > k$. 
Recall that $I_d(i, n)$ is the indicator function of the event
$\{\degree{G^{k}(n)}{v_i} = d\}$.  The total number of 
vertices of degree $d$ in $G^{k}(n)$ can thus be written as
$\tdeg{d}{n} = \sum\limits_{i = 1}^{n}I_{d}(i, n)$. 
By the additive property of conditional expectation, we have
\begin{equation}
\label{eq-cond-total}
\expectation{\tdeg{d}{n + 1} | \mathcal{F}_n}{} = \sum\limits_{i = 1}^{n + 1} \expectation{I_d(i, n + 1) | \mathcal{F}_n}{}.
\end{equation}
Due to the way in which $\rktree{k}{n}$ is constructed,  the vertex $v_i$ has degree $d$ in 
$\rktree{k}{n + 1}$ if and only if one of the following two situations occurs:
\begin{enumerate}
\item The degree of $v_i$ in $\rktree{k}{n}$ is $d$, and $v_i$  is not selected to be connected
to $v_{n+1}$; or
\item The degree of $v_i$ in $\rktree{k}{n}$ is $d - 1$, and $v_i$ is selected to be connected
to $v_{n+1}$.   
\end{enumerate}
Therefore,  letting $B$ be the event that $v_i$ is selected to be connected to $v_{n + 1}$, we have
\begin{equation} 
\label{eq:conde-indicator}
I_{d}(i, n + 1) = I_{B}I_{d - 1}(i, n) + I_{B^c}I_{d}(i, n).
\end{equation}
We claim that 
\begin{equation}
\label{eq:conde-g1}
\expectation{I_{B}I_{d - 1}(i, n) | \mathcal{F}_n}{}
 = \vcprob{k}{d-1}{n} I_{d-1}(i, n).
\end{equation}
We prove the claim by the mathematical definition of conditional expectation.
Consider any event $A \in \mathcal{F}_n$.  
(Recall that $I_{d - 1}(i, n)$ is the indicator function of the event
$\{\degree{\rktree{k}{n}}{v_i} = d - 1\}$.) We have 
\begin{eqnarray*}
&&\expectation{I_{B}I_{d - 1}(i, n)I_{A}}{} \\
    &=& \probab{B \cap \{\degree{\rktree{k}{n}}{v_i} = d - 1\} \cap A}{}  \\
    &=& \probab{B\ |\ \{\degree{\rktree{k}{n}}{v_i} = d - 1\} \cap A }{}
         \probab{\{\degree{\rktree{k}{n}}{v_i} = d - 1\}\cap A}{} \\
    &=& \probab{B\ |\ \{\degree{\rktree{k}{n}}{v_i} = d - 1\}}{}
        \probab{\{\degree{\rktree{k}{n}}{v_i} = d - 1\} \cap A}{} \\
    &=& \vcprob{k}{d-1}{n} \expectation{I_{d - 1}(i, n) I_{A}}{},           
\end{eqnarray*}
where the second last equality is due to the fact that the event 
$\{\degree{\rktree{k}{n}}{v_i} = d\}$ completely determines 
the (conditional probability of) the event $B$.
The claim then follows from the mathematical definition of conditional expectation.

Similarly, we have
\begin{equation} 
\label{eq:conde-g2}
\expectation{I_{B^c}I_{d}(i, n) | \mathcal{F}_n}{}
 = (1 - \vcprob{k}{d}{n}) I_d(i, n).
\end{equation}
 
Combining equations (\ref{eq:conde-indicator}), (\ref{eq:conde-g1}),
and (\ref{eq:conde-g2}), we see that
for any $i < n + 1$,
\begin{eqnarray}
\label{eq-cond-indicator}
&&\expectation{I_d(i, n + 1) | \mathcal{F}_n}{} = \vcprob{k}{d-1}{n}I_{d-1}(i, n) + (1 - \vcprob{k}{d}{n}) I_d(i, n).
\end{eqnarray}
Also note that for $i = n + 1$, by the construction of $\rktree{k}{n}$ we have 
$$
\expectation{I_{d}(n + 1, n + 1) | \mathcal{F}_n}{} = 0
$$ 
for any $d > k$.
Summing over $i$ on both sides of equation (\ref{eq-cond-indicator}) and based
on equation (\ref{eq-cond-total}), we have
\begin{equation}
\label{eq-recurr-1}
\expectation{\tdeg{d}{n + 1} | \mathcal{F}_n}{} = \vcprob{k}{d - 1}{n}\tdeg{d - 1}{n} + (1 - \vcprob{k}{d}{n})\tdeg{d}{n}.
\end{equation}
Recall that the unconditional expectation of the condition expectation of a random variable
 is equal to the unconditional expectation of the random variable itself. 
Taking unconditional  expectations  on both sides of equation (\ref{eq-recurr-1}),
we get the following recurrence equation for the expected number of vertices of degree $d$:
\begin{equation}
\label{eq-Y-recursion}
\edeg{d}{n + 1} = \vcprob{k}{d - 1}{n}\edeg{d - 1}{n} + (1 - \vcprob{k}{d}{n})\edeg{d}{n}.
\end{equation}
Using the recurrence equation (\ref{eq-Y-recursion}) and 
the base case equation (\ref{eq-formula-base}), we now prove
that 
$|\edeg{d}{n} - \beta_{d}n|$ is asymptotically upper bounded by a constant, 
where the  sequence \{$\beta_d\}$ is the
unique solution to the following simple recurrence equation      
\begin{equation}
\label{eq-limit-case}
 \beta_d = \frac{a_k(d - 1) - b_k}{a_kd - b_k + k} \beta_{d - 1},\ \  \beta_{k} = \frac{1}{2}.
\end{equation}
Let $\epsilon_d^n = \edeg{d}{n} - \beta_dn$. For the base case $d = k$, we have from 
equation (\ref{eq-formula-base}) that $\epsilon_k^{n}  =  O(1)$.
For the general case $d > k$, we have from equation (\ref{eq-Y-recursion}) that
\begin{eqnarray}
\label{eq:epsilon}
\epsilon_d^{n + 1} &=&  \vcprob{k}{d - 1}{n}\epsilon_{d - 1}^{n} + \vcprob{k}{d - 1}{n}\beta_{d - 1}n  \nonumber \\
                     \ \ \ \ \ \ \ \ \, \,  & & + (1 - \vcprob{k}{d}{n})\epsilon_{d}^{n} + (1 - \vcprob{k}{d}{n}) \beta_{d}n - (n+1)\beta_{d}.
\end{eqnarray}
By the definition of $\beta_d$ (equation (\ref{eq-limit-case})), we see that
\begin{eqnarray*}
 &&\vcprob{k}{d - 1}{n}\beta_{d - 1}n + (1 - \vcprob{k}{d}{n}) \beta_{d}n - (n+1)\beta_{d} \\
 && = \frac{a_k(d - 1) - b_k}{c_k}\beta_{d - 1} - \frac{a_kd - b_k}{c_k}\beta_d - \beta_d \\
  && = \frac{a_kd - b_k + k}{c_k}\beta_d - \frac{a_kd - b_k}{c_k}\beta_d - \beta_d  \ \ \ \ \ \ \ (\textrm{using (\ref{eq-limit-case}))}\\
  && = \frac{k - c_k}{c_k}\beta_d  \\
  && = \frac{k^2 - 1}{c_kn}\beta_d \ \ \ \ \ \ \ \ (\textrm{ since } c_k = k - (k^2 - 1)/n).
\end{eqnarray*}
Thus, we have
\begin{eqnarray}
\label{eq:epsilon:bound}
|\epsilon_d^{n + 1}| &\leq& \vcprob{k}{d - 1}{n} |\epsilon_{d - 1}^{n}| + (1 - \vcprob{k}{d}{n})|\epsilon_d^{n}|
                           + \frac{k^2 - 1}{nc_k} \beta_{d} \nonumber \\
                  &\leq& (1 + \vcprob{k}{d - 1}{n} - \vcprob{k}{d}{n}) \max(|\epsilon_{d - 1}^{n}|, |\epsilon_d^{n}|)
                    + \frac{k^2 - 1}{nc_k} \beta_{d}
\end{eqnarray}
From (\ref{eq-limit-case}), we see that $\beta_d \leq \beta_{d - 1}$ for any $d > k$.
Since  by (\ref{eq:epsilon}) $\epsilon_k^{n}  =  O(1)$ and since 
$$
\vcprob{k}{d - 1}{n} - \vcprob{k}{d}{n} = -\frac{a_k}{nc_k} < 0
$$
by the definition of  $\vcprob{k}{d}{n}$, we can use (\ref{eq:epsilon:bound}) to prove by induction
 that there exists a constant $N = N(k) > 0$ independent of $d$ such that for any $n > N$,
$|\epsilon_{d}^n|$ is bounded by a constant $C = C(k)$ independent of $d$ and $n$, and therefore
\begin{equation}
\label{eq-claim}
  \edeg{d}{n} = \beta_d n + O(1).
\end{equation} 
To complete the proof of Theorem~\ref{theorem-average},
we see from the definition of $\beta_d$ that
\begin{eqnarray}
\beta_d &=& \prod\limits_{l = k}^{d}\frac{a_k(l - 1) - b_k}{a_kl - b_k + k}
  = \prod\limits_{l = k}^{d}\frac{l - 1 - \frac{b_k}{a_k}}{l - \frac{b_k}{a_k} + \frac{k}{a_k}} \nonumber \\
  &=&\frac{\Gamma(k - \frac{b_k}{a_k} + \frac{k}{a_k} + 1)}{\Gamma(k - \frac{b_k}{a_k})} \frac{\Gamma(d - \frac{b_k}{a_k})}{\Gamma(d - \frac{b_k}{a_k} + \frac{k}{a_k} + 1)} \nonumber \\
  &=&\frac{\Gamma(3 + \frac{2}{k - 1})}{\Gamma(1 + \frac{1}{k - 1})} \frac{\Gamma(d - \frac{b_k}{a_k})}{\Gamma(d - \frac{b_k}{a_k} + \frac{k}{a_k} + 1)}
\end{eqnarray}
which by Stirling's approximation is
approximately  $e^{-(1 + k / a_k )}d^{-(1 + k / a_k )}$ for large $d$.
\end{proof}

Next, we show that $\tdeg{d}{n}$, the number of vertices of degree $d$, concentrates on its expectation,  which together with
Theorem \ref{theorem-average}, establishes the power law degree distribution of the \ktree\ process.

\begin{theorem}
\label{theorem-proportion}
Let $\tdeg{d}{n}$ be the total number of vertices of degree $d$ in $\rktree{k}{n}$. For any $\lambda > 0$, we have
\begin{equation}
\label{eq-proportion}
\probab{|\tdeg{d}{n} - \expectation{\tdeg{d}{n}}{}| > \lambda}{} \leq  e^{-\frac{\lambda^2}{8kn}}{}.
\end{equation}
\end{theorem}
\begin{proof}
Consider the martingale $\{Z_i \triangleq \expectation{\tdeg{d}{n} | \mathcal{F}_i}{}, i \geq k + 1\}$ and the associated martingale difference sequence
$ \{Z_{i + 1} - Z_{i},  i \geq k + 1\}.$
If we can show that
$$
    |Z_{i + 1} - Z_{i}| \leq 2k,
$$
then an application of Azuma's Inequality (see, e.g. Theorem 7.4.2 of \cite{alon00}) gives   (\ref{eq-proportion}).

For each $i$, let
$\mathcal{C}_i$ be the collection of size-(k+1) vertex subsets  of the form $\{v_{i_1}, \cdots, v_{i_{k}},  v_{i}\}$ where
$i_1 < i_2 < \cdots < i_{k} < i$. $\mathcal{C}_i$ is the collection of the possible (k+1)-cliques in $\rktree{k}{i}$ that
contain $v_i$ as one of their vertices. We call $v_i$ the head of a subset $C \in \mathcal{C}_{i}$ and write $head(C) = v_i$.

Now consider the probability space $(\Omega, \probab{\cdot}{})$ defined over the product space
$\Omega = \prod\limits_{i = k + 1}^{n} \mathcal{C}_i$. A sample point $h = \{h_i\} \in \Omega$ is
said to be a \textit{realization} of a \ktree\ if for any $k + 1 \leq i \leq n$, the vertex subset
$h_i = \{v_{i_1}, \cdots, v_{i_k}, v_i\} \in \mathcal{C}_i$ is such that $\{v_{i_1}, \cdots, v_{i_k}\}$ is a subset
of $h_j$ for some $j < i$.

Let $\Omega_{0} \subset \Omega$ be the set of sample points that are realizations of a \ktree.
The probability measure $\probab{\cdot}{}$ is defined as follows. It
has $\Omega_{0}$ as its support and for each $h\in \Omega_{0}$,
$$
  \probab{h}{} \stackrel{\text{def}}{=}  
  \prod\limits_{i = k + 1}^{n - 1} \frac{1}{k + 1}\frac{1}{i - k}.
$$
The reason for $\probab{h}{}$, where $h = (h_i)$, to be defined as in the above is explained  
as follows. Let $G$ be the $k$-tree on the vertex set $\{v_1, \cdots, v_n\}$ such that
for each $k+2 \leq i\leq n$, when $v_i$ is added, it is connected to
a subset of $k$ vertices from some $h_j$ where $k + 1 \leq j < i$.    
The probability that the random $k$-tree $\rktree{k}{n}$ obtained according to our construction
is equal to $G$ is
$$
\probab{\rktree{k}{n} = G}{} =
 \prod\limits_{i = k + 1}^{n - 1} \frac{1}{k + 1}\frac{1}{i - k},
$$ 
where the term $\frac{1}{k + 1}\frac{1}{i - k}$ is the conditional probability 
(given $\rktree{k}{i - 1}$) that a specific size-$k$ vertex subset in a specific $h_j$
is selected to be connected to $v_{i + 1}$.

For any $g \in \Omega_0$, writing  $\tdeg{d}{n, g}$ for the total number of vertices of degree $d$
in the \ktree\ realized by $g$, we have
$$
  Z_{i}(h) =
  \expectation{\tdeg{d}{n, g} | g_j = h_j,  \forall k + 1\leq j \leq i}{}.
$$

The following argument is motivated by a similar one used in \cite{cooper03}.
Let $h = \{h_{k + 1}, \cdots, h_n\} \in \Omega_0$ be a sample point and
$H_{i+1}[h] \subset \Omega_0$ be the collection of sample points that agree with
$h$ for $k+1\leq j \leq i+1$, i.e.,
$$
 H_{i+1}[h] = \{g\in \Omega_0|\ \ g_j = h_j, \forall k+1\leq j \leq i+1\}.
$$

Consider a size-(k+1) vertex set $C \in \mathcal{C}_{i+1}$ such that
$C \neq h_{i + 1}$. Define $H_{i+1}[C]$ to be the collection of the sample points
$g \in \Omega_0$ such that
\begin{equation}
\left \{
\begin{array}{ll}
&g_j = h_j, \forall k+1\leq j \leq i \\
&g_{i+1} = C
\end{array}
\right.
\end{equation}
We claim that there is a one-to-one correspondence between the elements of
$H_{i+1}[h]$ and $H_{i+1}[C]$.

Assume that $h_{i+1} = \{v_{j_{1}}, \cdots, v_{j_k}, v_{i+1}\}$ and
$C = \{v_{l_1}, \cdots, v_{l_k}, v_{i + 1}\}$.
The claimed one-to-one correspondence
can be shown by the mapping defined as follows. For each $1\leq p \leq k$, define
$\sigma_{i+1}(v_{j_p}) = v_{l_p}$. For each $g \in H_{i+1}[h]$, define $\sigma(g) =
g' \in H_{i+1}[C]$ as
\begin{enumerate}
\item $g'_j = h_j$ for any $j \leq i$;
\item $g'_{i + 1} = C $; and
\item for each $j > i + 1$, $g'_{j}$ is a size-(k+1) vertex subset defined as
\begin{enumerate}
\item if $g_j$ doesn't contain the vertex $v_{i + 1}$, then $g'_j = g_j$, and
\item if $g_j$ contains $v_{i + 1}$, then $g'_j$ is obtained by replacing
 each vertex $v \in g_j \cap h_{i + 1}$ with $\sigma_{i + 1}(v)$.
\end{enumerate}
\end{enumerate}
For any $g \in H_{i + 1}[h]$, since the only vertices whose vertex degree might have been changed by the
mapping $\sigma$ are those
in $\{v_{j_1}, \cdots, v_{j_k}\}$ and $\{v_{l_1}, \cdots, v_{l_k}\}$, we have
$$
  |\tdeg{d}{n, g} - \tdeg{d}{n, \sigma(g)}| \leq 2k.
$$

Since the probability measure $\probab{\cdot}{}$ assigns equal probability  to the sample points,
$$
    |Z_{i + 1} - Z_{i}| \leq 2k, \forall\ h\in \Omega_0,
$$
holds due to the definition of conditional expectation:
$$
 Z_{i + 1} = \expectation{X(n, h) | \mathcal{F}_{i+1}}{}
  = \sum\limits_{g\in H_{i+1}[h]} X(n, g) \probab{g | g_j = h_j, j \leq i + 1}{}
$$
and
$$
 Z_{i} =  \expectation{X(n, h) | \mathcal{F}_{i}}{}
  = \sum\limits_{C \in \mathcal{C}_{i+1}}
     \sum\limits_{g\in H_{i+1}[C]} X(n, g) \probab{g | g_j = h_j, j \leq i}{}.
$$
This completes the proof.
\end{proof}

\section{Concluding Remarks}
In this paper, we have shown that a simple evolving graph model based on the the notion of $k$-trees has a power law degree distribution with high probability. Due to its simplicity and unique structures, we think this model of evolving graphs provides a useful alternative in the modeling, analysis, and simulations of complexity networks such as the web graphs that have attracted much attention \cite{chakrabarti06}.  The fact that    
a power law degree distribution exists in such a graph models with quite unique characteristics
also serves as a further cautionary note, reinforcing the viewpoint that a power law degree 
distribution should not be regarded as the only important feature of a complex network, as has been previously argued in \cite{dimitris05power,li05,mitzenmacher05}.

A partial \ktree\ is a subgraph of a \ktree. To enrich the modelling power of the class of models, it is desirable to have a natural model of
random partial \ktrees.  It is tempting to think of the following model based on the construction of the \ktree\ process:  For each vertex $v$ in
$\rktree{k}{n}$, delete randomly-selected $(1 - b) * k$ of its $k$ out-edges for some 
$0 < b < 1$. In \cite{gao06waw}, we claimed that a model of random partial \ktrees\ obtained in this way has
a power law degree distribution $d^{-(1 + \frac{k}{b*(k-1)})}$, which turns out to be flawed. A few alternatives have since then been
investigated, resulting in very unnatural random models. We leave it as an open question the existence of a natural evolution model for the partial \ktrees.

\section*{Acknowledgements} 
The author would like to thank the anonymous referees for their constructive comments, and 
Christopher Hobson for his help in conducting simulations on a random partial k-tree model.
\bibliographystyle{plain}
\bibliography{../treewidth,../complex_networks,../random_graph}

\end{document}